\documentclass{IEEEtran}  % FINAL VERSION (12pp)

\IEEEoverridecommandlockouts                              % This command is only
\def\TITLEnl{Distribution Network Topology Detection with Time Series Measurement Data Analysis}
\def\AUTHORS{Guido Cavraro, Reza Arghandeh, Alexandra von Meier}

\usepackage[cmex10]{amsmath}
\usepackage{mdwmath}
\usepackage{amsbsy}
\usepackage{eqparbox}
\usepackage{multirow}

\usepackage{algorithm}							
\usepackage{algpseudocode}

\usepackage{amsmath}
\usepackage{amsfonts}
\usepackage{amssymb}
\usepackage{url}

\newtheorem{lemma}{Lemma}
\newtheorem{proposition}{Proposition}
\newtheorem{remark}{Remark}

\newtheorem{assumption}{Assumption}
\newtheorem{example}{Example}

\usepackage{enumerate}

\usepackage{array}

\usepackage[english]{babel}

\usepackage{graphicx} 
\usepackage{color}
\usepackage{lineno, blindtext}

\renewcommand{\baselinestretch}{1}

\def\complexnumbers{\mathbb{C}}

\def \graph	{\mathcal{G}}		% grafo di comunicazione
\def \nodes	{\mathcal{V}}		% insieme dei nodi
\def \adm {{\boldsymbol Y}}

\def \green {{\boldsymbol X}}
\def \edges	{\mathcal{E}}		% insieme delle coppie di nodi comunicanti

\def \topology {{\boldsymbol T}}

\def \nonodes {{n}}
\def \noswitches {{r}}
\def \nosensors {{p}}
\def \noedges {{w}}

\def\1{{\mathbf{1}}}
\title{\LARGE \bf \TITLEnl}

\author{\AUTHORS%
\thanks{G. Cavraro is with the Department of Information Engineering, University of Padova, Italy. 
Email: {\tt\small cavraro@dei.unipd.it.} %
}
\thanks{R. Arghandeh is with the California Institute for Energy and Environment (CIEE), University of California, Berkeley, CA, USA. Email: {\tt\small arghandeh@berkeley.edu}}% <-this % stops a space
}

\begin{document}

\maketitle

\begin{abstract}
This paper proposes a novel approach for detecting the topology of distribution networks based on the analysis of time series measurements. The time-based analysis approach draws on data from high-precision phasor measurement units (PMUs or synchrophasors) for distribution systems. A key fact is that time-series data taken from a dynamic system show specific patterns regarding state transitions such as opening or closing switches, as a kind of signature from each topology change. The proposed algorithm here is based on the comparison of the actual signature of a recent state transition against a library of signatures derived from topology simulations. The IEEE 33-bus model is used for initial algorithm validation. 
\end{abstract}

%\IEEEpeerreviewmaketitle
%\begin{linenumbers}
\section{Introduction}
Different tools have been developed and implemented to monitor distribution network behavior with more detailed and temporal information, such as SCADA, smart meters and line sensors. Creating observability out of disjointed data streams still remains a challenge, though. Given the present monitoring technologies, more, better and faster data from behind the substation will be needed to realize smart distribution networks \cite{vonMeier2014Chap}.
The cost of monitoring systems in distribution networks remains a barrier to equip all nodes with measurement devices. To some extent, a capable Distribution System State Estimation can compensate for the lack of direct sensor data to support observability. 
However, switches status errors can easily be misinterpreted as analog measurement errors (e.g. voltage or current readings). Thus topology detection is an important enabling component for state estimation as well as a host of other operation and control functions based on knowledge of the system operating states in real-time.

Topology detection in distribution networks is not a trivial problem. Firstly, switches may not reliably communicate their status to the distribution operator, so the topology can only be determined by sending crews into the field. Moreover, the reported switch status may be faulty due to switch malfunction and unreported maintenance crew manipulation.
Integration of distributed energy resources, electric vehicles and controllable loads add more dynamics to the distribution networks that may lead to more frequent protection and switching actions.
A survey of utilities experts done by the authors shows on average 5 to 10 switching actions happens under an urban distribution substation \cite{AMP2014}. Furthermore, knowledge of the correct, updated topology is essential for safety and service restoration after outages\cite{Lueken2012}. 

Most of literature on topology detection are based on state estimator results and measurement matching with different topologies. In \cite{korres2012} authors propose a state estimation algorithm that incorporates switching device status as additional state variables. A normalized residual test is used to identify the best estimate of the topology. State estimation based algorithms are limited to state estimator accuracy. They are also sensitive to measurement device placement. In \cite{sharon2012}, the authors provide a tool for choosing sensor placement for topology detection. Given a particular placement of sensors, the tool reveals the confidence level at which the status of switching devices can be detected. Authors in \cite{Ciobotaru2007} and \cite{Liserre2007} are focused on estimating the impedance at the feeder level. However, even a perfect identification of network impedance cannot always guarantee the correct topology, since multiple topologies could present very similar impedances.

The work in \cite{Arya2011} used power flow analysis for matching substation loading and aggregated household meter load data from network connectivity modeling.  They assumed metering load data are time synchronized with a measurement device on each transformer. The assumption is still far from an actual load metering system. Moreover, convergence of the proposed optimization is sensitive to bad data. Voltage measurement cross correlation for house meters data is presented in \cite{Short2013}. However, residential meters do not usually provide voltage measurements for utility operations. Voltage correlation-based methods in balanced feeders, feeders with PV resources, and feeders with inaccurate GIS models can be error prone. Moreover, voltage measurement in meters are hourly or 15 minute average values. Voltage average values introduce additional errors in voltage regression-based methods for topology detection. Most of the proposed methods in literature are post-processing methods which depend on correct execution of  state estimation or power flow.

In this paper, a novel approach for topology detection is proposed based on time series analysis of phasor measurement unit (PMU) data. This approach is inspired by high-precision phasor measurement units for distribution systems micro-synchrophasors ($\mu$-PMU), which the authors are involved in implementing \cite{vonMeier2014}. The main idea derives from the fact that time-series data from a dynamic system show specific patterns regarding system state transitions, a kind of signature left from each topology change. 
The algorithm is based on the comparison of the trend vector, built from system observations, with a library of signatures derived from the possible topologies transitions. The topology detection results are impacted by load uncertainty and measurement device accuracy. Therefore, the conduced analysis takes load dynamics and measurement error into account.
A set of actual secondly load measurements are used for a number of residential customers in the United States. The derived statistical load model is applied to the topology detection scenarios to validate the proposed algorithm. The topology detection accuracy is also dependent on  $\mu$-PMU placement. We proposed a $\mu$-PMU placement approach for topology detection application. The analysis shows that the topology detection algorithm converges even with limited measurement devices. 
%The proposed algorithm is measurement data-driven approach so it is adaptable to any network with different measurement systems.

%The rest of this paper is organized as follows: Section \ref{sec:model} describes the distribution network model. Section \ref{sec:propagation_action_switches} describes how switching actions will propagate in the mathematical representation. Section \ref{sec:algorithm_statement} presents our topology detection algorithm, and Section \ref{sec:results} shows the initial validation through simulation in a 33-bus system.

\section{Distribution Grid Model}
\label{sec:model}
This section presents the distribution network model and its related notations. 
Given a matrix $W$, we denote its transpose by $W^T$  and its conjugate transpose by $W^*$. We denote the matrices of the real and of the imaginary part of $W$ by $\Re(W)$ and by $\Im(W)$, respectively. We denote the entry of $W$ that belongs to the $j$-th row and to the $k$-th column by $W_{jk}$. 

Given a vector $v$, $v_j$ will denote its $j$-th entry, while $v_{-j}$ the subvector of $v$, in which the $j$-th entry has been eliminated. We denote its complex conjugate by $\overline{v}$. Given two vectors $v$ and $w$, we denote by $\langle v,w \rangle$ their inner product $v^*w$. 
We define the column vector of all ones by $\1$.

We associate with the electric grid the directed graph $\graph= (\nodes, \edges)$ and the sets $\mathcal S$ and $\mathcal P$, where 
\begin{itemize}
\item $\nodes$ is the set of nodes (the buses), with cardinality $\nonodes$;
\item $\edges$ is the set of edges (the electrical lines connecting the buses), with cardinality $\noedges$; 
\item $\mathcal S$ is the set of switches (or breakers) deployed in the electrical grid, with cardinality $\noswitches$
\item $\mathcal P$ is the set of nodes endowed with voltage phasor measurement units (PMUs), with cardinality $\nosensors$
\end{itemize}
Let $A \in \{0, \pm 1\}^{\noedges \times \nonodes}$ 
be the incidence matrix of the graph $\graph$, 
$$A=
\begin{bmatrix}
a_1^T &
\dots &
a_\noedges^T
\end{bmatrix}^T
$$
where $a_j$ is the $j$-th row of $A$. The elements of $A$ are all zeroes except for the entries associated with the nodes connected by the $j$-th edge, for which they are equal to $+1$ or $-1$  respectively.  
%If the graph $\graph$ is connected (i.e. for every pair of nodes there is a path connecting them), then $\1$ is the only vector in the null space of $A$ ($\ker A$), $\1$ being the column vector of all ones.
In this study, we consider the steady state behavior of the system, when all voltages and currents are sinusoidal signals waving at the same frequency $\omega_0$. Thus, they can be expressed via a complex number whose magnitude corresponds to the signal root-mean-square value, and whose phase corresponds to the phase of the signal with respect to an arbitrary global reference. Therefore, $x$ represents the signal
$
x(t) = |x| \sqrt{2} \sin(\omega_0 t + \angle x).
$
The system state is described by the following quantities:
\begin{itemize}
\item $u \in \complexnumbers^\nonodes$, where $u_v$ is the grid voltage at node $v$;
\item $i \in \complexnumbers^\nonodes$, where $i_v$ is the current injected at node $v$;
\item $s = p + i q \in \complexnumbers^\nonodes$, where $s_v$, $p_v$ and $q_v$ are the complex, the active and the reactive power injected at node $v$, respectively;
\item $\sigma \in \{0,1\}^ \noswitches$, where $\sigma_v$ is the staus of the breaker $v$: $\sigma_v = 0$ if the switch $v$ is open, $\sigma_v = 1$ if the switch $v$ is closed;
\item $y \in \complexnumbers^\nosensors$, where $y_v$ is the grid voltage measured by the sensor $v$ 
\item the \emph{trend vector} $\delta(t_1,t_2)\in \complexnumbers^\nosensors$, defined as the difference between phasorial voltages taken at the two time istant $t_1$ and $t_2$. i.e. $ \delta(t_1,t_2)={u}(t_1)-{u}(t_2)$.
\end{itemize}
We assume that the PMUs deployed in the grid take measurements at the frequency $f$. 
We denote with $\topology_\sigma$ the unique topology whose switches status is described by $\sigma$. Its bus admittance matrix $\adm_\sigma$ is defined as
\begin{equation}
(\adm_\sigma)_{jk} = \begin{cases}
\sum_{j \neq k} Y_{jk}, \text{ if } j = k \\
- Y_{jk}, \text{ otherwise }
\end{cases}
\label{eq:busAdmMatrixDef}
\end{equation}
where $Y_{jk}$ is the admittance of the branch connecting bus $j$ and bus $k$ and where we neglect the shunt admittances. From \eqref{eq:busAdmMatrixDef}, it can be seen that $\adm_\sigma$ is symmetric and satisfies
\begin{equation}
\adm_\sigma \1 = 0,
\label{eq:admittance_null_space}
\end{equation}
i.e. $\1$ belongs to the Kernel of $\adm_\sigma$. Furthermore, it can be shown that if $\graph$, the graph associated to the electrical grid, is connected, then the kernel of $\adm_\sigma$ has dimension 1.

We model the substation as an ideal sinusoidal voltage source (\emph{slack bus}) at the distribution network nominal voltage $U_N$. We  assume, without loss of generality, that $U_N$ is a real number. We model all nodes but the substation as \emph{constant power devices}, or \emph{P-Q buses}. In the following, $u_1, s_1, p_1, q_1$ will denote the voltage of the substation, and its complex, its active and its reactive power injected, respectively.
%
%The powers $s_v$ corresponding to grid loads are such that $p_v, q_v <0$, meaning that positive active power is \emph{supplied} to the devices. On the other hand, the complex powers corresponding to microgenerators are such that $p_v \ge 0$, as positive active power is \emph{injected} into the grid.
The system state satisfies the following equations
\begin{align}
&i = \adm_\sigma u \label{eq:nodevoltage}\\
&u_1 = U_N  \label{eq:PCCidealvoltgen}\\
&u_v \bar i_v = p_v + i q_v \qquad v\neq 0 \label{eq:nodeconstpwr} 
\end{align}

The following Lemma \cite{Zampieri} introduces a particular, very useful pseudo inverse of $\adm_\sigma$.
\begin{lemma}
There exists a unique symmetric, positive semidefinite matrix $\green_\sigma \in \complexnumbers^{n\times n}$ such that
\begin{equation}
\begin{cases}
\green_\sigma \adm_\sigma = I - \1 e_1^T \\
\green_\sigma e_1 = 0.
\end{cases}
\label{eq:Xproperties}
\end{equation}
\label{lemma:X}
\end{lemma}
Applying Lemma \ref{lemma:X}, from \eqref{eq:nodevoltage} and \eqref{eq:PCCidealvoltgen} we can express the voltages as a function of the currents and of the nominal voltage $U_N$:
\begin{equation}
u=\green_\sigma i+\1 U_N
\label{eq:u=Xi}
\end{equation}
The following proposition (\cite{Zampieri}) provides an approximation of the relationship between voltages and powers.
\begin{proposition}
Consider the physical model described by the set of nonlinear equations \eqref{eq:nodevoltage}, \eqref{eq:PCCidealvoltgen},  \eqref{eq:nodeconstpwr} and \eqref{eq:u=Xi}.
Node voltages then satisfy
\begin{equation}
u
=
U_N \1 + 
\frac{1}{U_N}
\green_\sigma
\bar s
+
o\left(\frac{1}{U_N}\right)
\label{eq:approximate_solution}
\end{equation}
(the little-o notation means that $\lim_{U_N\rightarrow \infty} \frac{o(f(U_N))}{f(U_N)} = 0$).
\label{pro:approximation}
\end{proposition}
Equation \eqref{eq:approximate_solution} is derived basically from a first order Taylor expansion w.r.t. the nominal voltage $U_N$ of the relation among powers and voltages. Even if the solutions of \eqref{eq:approximate_solution} could be not very precise, and then its application to the power flow computation could be not satisfactory, it has been already used with success in state estimation \cite{schenato2014bayesian},  Volt/Var optimization \cite{cavraroTAC2015}, and the optimal power flow problem \cite{cavraro2014cdc}. 
%
%
%%%%%%%%%%%%%%%%%%%%%%%%%%%%%%%%%%%%%%%%%%%%%%%%%%%%%%%%%%%%%%%%%%%%%%%%%%%%%%%%%%
%%%%%%%%%%%%%%%%%%%%%%%%%%%%%%%%%%%%%%%%%%%%%%%%%%%%%%%%%%%%%%%%%%%%%%%%%%%%%%%%%%
%

\section{Identification of Switching Actions }
\label{sec:propagation_action_switches}

The basic idea behind our proposed approach is that changes of the breakers status will create specific signatures in the grid voltages.
In order to develop the theoretical base for the proposed algorithm, we make the following assumptions. 
%
%\begin{assumption}
%\label{ass:not_time_var_loads}
%We assume that the loads are not time varying.
%\end{assumption}
%
\begin{assumption}
\label{ass:sameR/X}
We assume that all the lines have the same resistance over reactance ratio. Therefore, $\Im(Y_{jk}) = \alpha \Re(Y_{jk}), \forall \; Y_{jk}$.
\end{assumption}
\begin{assumption}
\label{ass:statuschng}
We assume that only one switch can change its status at each time.
\end{assumption}
\begin{assumption}
\label{ass:connection}
We assume that the graph associated with the grid is always connected, i.e. that there are no admissible state in which any portion of the grid remains disconnected.
\end{assumption}
%It follows, from assumption \ref{ass:sameR/X}, that we can write 
%$$\green = e^{i \theta} X, \adm = e^{-i \theta} Y.$$
%
Assumption \ref{ass:sameR/X} will be relaxed in Section \ref{sec:results}, in order to test the algorithm in a more realistic scenario. However, it allows us to write the bus admittance matrix in a simpler way. In fact, due also to $\adm_{\sigma}$ symmetry,
  %$\Im(\adm_{\sigma})$ and $\Re(\adm_{\sigma})$ share the same eigenvectors. Therefore, 
	we can write
\begin{equation}
\adm_{\sigma} = \Re(\adm_{\sigma}) + i \Im(\adm_{\sigma})  = ( 1 + i \alpha) U \Sigma_R U^*
\end{equation}
where $\Sigma_R$ is a diagonal matrix whose diagonal entries are the non-zero eigenvalues of $\Re(\adm_{\sigma})$, and $U$ is an orthonormal matrix that includes all the associated eigenvector. From \eqref{eq:admittance_null_space}, it can be shown that $U$ spans the image of $I - \1 \1^T/ (\1^T \1)$, i.e. the space orthogonal to $\1$. The matrix $\green_\sigma$ can be written as 
\begin{equation}
\green_\sigma = (1 + i \alpha)^{-1} \Lambda U \Sigma_R^{-1} U^* \Lambda^T
\label{eq:MoorePenros}
\end{equation}
with $\Lambda = (I - \1 e_1^T)$. Assumption \ref{ass:statuschng} is instead reasonable for the proposed algorithm framework: it works in the scale of some seconds, while typically the switches are electro-mechanical devices and their actions are not simultaneous. Finally, Assumption \ref{ass:connection} is always satisfied during the normal operation.
We introduce the main idea that underlies our algorithm with the following example.
\begin{example}
\label{ex:es1}
Assume that at time $t-1$ the switches status is described by $\sigma(t-1) =  \sigma_1$, resulting in the topology $\topology_{\sigma_1}$ with bus admittance matrix $\adm_{\sigma_1}$. Applying Proposition~\ref{pro:approximation} and neglecting the infinitesimal term, we can express the voltages as
\begin{equation}
{u}(t-1)=
\green_{\sigma_1} 
\frac{\bar s }{U_N} + \1 U_N
\label{eq:ux1}
\end{equation}
At time $t$ the $\ell$-th switch, that was previously open, changes its status. Let the new status be described by $\sigma(t) =  \sigma_2$, associated with the topology $\topology_{\sigma_2}$. Since we are basically adding to the graph representing the grid the edge on which switch $\ell$ is placed, we can write 
\begin{equation}
\adm_{\sigma_2} = \adm_{\sigma_1} + Y_{\ell} a_{\ell} a_{\ell}^T
\end{equation}
where  $Y_{\ell}$ is the admittance of the line, and the elements of $a_{\ell}$ is the $\ell$-th row of the adjacency matrix associated with $\topology_{\sigma_2}$. %
Since $ a_{\ell} $ is orthogonal to  $\1$, there exists $b_{\ell}$ such that $Ub_{\ell} = a_{\ell}$. This allow us to write
\begin{small}
\begin{align}
\adm_{\sigma_2} &= (1 + i\alpha)  U (\Sigma_R + \Re(Y_{\ell}) b_{\ell} b_{\ell}^T) U^* \notag\\
\green_{\sigma_2} &= (1 + i\alpha)^{-1} \Lambda U (\Sigma_R + \Re(Y_{\ell}) b_{\ell} b_{\ell}^T)^{-1} U^* \Lambda^T \label{eq:X2}
\end{align}
\end{small}
The voltages are
\begin{equation}
{u}(t) =
\green_{\sigma_2} \frac{\bar s }{U_N}
+ \1 U_N
\label{eq:ux2}
\end{equation}
The trend vector $\delta(t,t-1)$ can thus be written as
\begin{equation}
\delta(t,t-1) = (\green_{\sigma_2} - \green_{\sigma_1}) \frac{\bar s }{U_N}
\label{eq:Delta1}
\end{equation}
\end{example}
\begin{example}
Consider now the opposite situation:  at time $t-1$ the $\ell$-th switch is closed and it is opened at time $t$, i.e. $\sigma(t-1) = \sigma_2, \sigma(t) = \sigma_1$. In this case we have that 
\begin{align}
\delta(t,t-1) & =  (\green_{\sigma_1} - \green_{\sigma_2}) \frac{\bar s }{U_N} \label{eq:Delta2}
\end{align}
\end{example}

We can observe that when there is a switching action, the voltage profile varies in a way determined by the particular transition from a topology to another. %In particular, the trend vector $\delta(t_1,t_2)$ depends on $\Phi_{_{\sigma(t_1)\sigma(t_2)}}$. 
Furthermore, opening or closing a determined breaker makes the system vary in the opposite way, e.g. compare \eqref{eq:Delta1} with \eqref{eq:Delta2}.
This is basically due to the fact that $[\sigma_1]_{-\ell} = [\sigma_2]_{-\ell}$. 

Given a breaker status $\sigma$ 
%that satisfies Assumption \ref{ass:connection}, 
associated with the topology $\adm_\sigma = (1 + i \alpha) U \Sigma_R U^*$, and fixed an edge $\ell$ endowed with a breaker, associated with the row $a_\ell$ in the incidence matrix,  we define the \emph{signature matrix} $\Phi_{\sigma_{-\ell}}$ as 
\begin{equation}
\Phi_{\sigma_{-\ell}} = U \Sigma_R^{-1} U^* - U (\Sigma_R + \Re(Y_{\ell}) b_{\ell} b_{\ell} ^T)^{-1} U^*
\end{equation}
Exploiting the signature matrix we can write
\begin{equation}
\delta(t,t-1) = \Lambda \Phi_{\sigma(t)_{-\ell}} \Lambda^T \frac{\bar s }{U_N}
\label{eq:closesw}
\end{equation}
if the switch $\ell$ has been closed, else 
\begin{equation}
\delta(t,t-1) = - \Lambda \Phi_{\sigma(t)_{-\ell}} \Lambda^T \frac{\bar s }{U_N}
\label{eq:opensw}
\end{equation}
if the switch $\ell$ has been opened. 
The following Proposition show a characteristic of the signature matrix that is fundamental for  the development of our topology detection algorithm.
\begin{proposition}
For every transition from the state described by $\sigma(t-1)$ to the one described by $\sigma(t)$, due to a action of the switch $\ell$, $\Phi_{\sigma(t)_{-\ell}}$ is a rank one matrix.
\label{pro:Phi_rank}
\end{proposition}
\begin{proof}
Following the same reasoning of Example \ref{ex:es1}, and exploiting %\eqref{eq:X1}, \eqref{eq:X2} and using 
K. Miller Lemma (\cite{miller1981inverse}), after some simple computations we can write
\begin{equation}
\Phi_{\sigma(t)_{-\ell}} = \mu U \Sigma_R^{-1} b_{\ell} b_{\ell}^T \Sigma_R^{-1} U^*
\label{eq:PhiMiller}
\end{equation}
with 
$$\mu = \frac{1}{1 + \text{Tr}(\Re(Y_\ell) b_{\ell} b_{\ell}^T\Sigma_R)}$$
It's trivial to see that $\Phi_{w(t)_{-\ell}}$ is a rank one matrix with eigenvector 
\begin{equation}
\hat g_{\sigma(t)_{-\ell}} = U \Sigma_R^{-1} b_{\ell} 
\label{eq:eigvct}
\end{equation}
associated with the non-zero eigenvalue
\begin{equation}\lambda_{\sigma(t)_{-\ell}} = \mu \|U \Sigma_R^{-1} b_{\ell} \|^2
\end{equation}
%
%We can therefore write
Thus, we have
$$\Phi_{\sigma(t)_{-\ell}} = \lambda_{\sigma(t)_{-\ell}} \hat g_{\sigma(t)_{-\ell}} \hat g_{\sigma(t)_{-\ell}}^*$$
%being all the information included in a rank one matrix contained in its non-zero eigenvalue and its associated eigenvector.
\end{proof}
The trend vector $\delta(t,t-1)$ represents how the opening or the closure of a switch spreads on the voltages profile. 
Thanks to Proposition \ref{pro:Phi_rank} we can write it as
\begin{align}
\delta(t,t-1) &= %\Lambda \Phi_{\sigma(t)_{-\ell}} \Lambda^T \frac{\bar s}{U_N}  \notag  \\
%&= \Lambda \lambda_{\sigma(t)_{-\ell}} \hat g_{\sigma(t)_{-\ell}} \hat g_{\sigma(t)_{-\ell}}^* \Lambda^T \frac{\bar s}{U_N}  \notag \\
%&
= \left[\lambda_{\sigma(t)_{-\ell}}  \hat g_{\sigma(t)_{-\ell}}^* \Lambda^T \frac{\bar s}{U_N} \right] \Lambda \hat g_{\sigma(t)_{-\ell}} 
\label{eq:trend_prop}
\end{align}
from which we see that 
%
%\begin{equation}
$\delta(t,t-1) \propto \Lambda \hat g_{\sigma(t)_{-\ell}},$
%\label{eq:trend_prop}
%\end{equation} 
%
i.e., every pattern that appears on the voltage profile due to  switching actions is proportional to $\Lambda \hat g_{\sigma(t)_{-\ell}}$, irrespective of other variables such as voltages $u$ and loads $s$ that describe the network operating state at the time. Thus, $g_{\sigma(t)_{-\ell}}$ can be seen as the \emph{particular signature} of the switch action. This fact is the cornerstone for the topology detection algorithm in this paper.
%
%
%\begin{remark}
%Equation \eqref{eq:trend_prop} point out the fact that the trend vector after a switch action will always be proportional to $g_{\sigma(t)_{-\ell}}$, irrespective of other variables such as voltages $u$ and loads $s$ that describe the network operating state at the time.  
%\end{remark}
%
\begin{remark}
\label{rem:2}
The opening and the closure of the switch $\ell$, once the other switches status is fixed, share the same signature $g_{\sigma(t)_{-\ell}}$, thus in principle they are indistinguishable. Without any other information, the trend vector can just identify witch switch has changed its status.
\end{remark}

%
%%%%%%%%%%%%%%%%%%%%%%%%%%%%%%%%%%%%%%%%%%%%%%%%%%%%%%%%%%%%%%%%%%%%%%%%%%%%%%%%%%%%%%%%%%%%%%%%%%%%%%%%%%%%%%%%%%%%%
%%%%%%%%%%%%%%%%%%%%%%%%%%%%%%%%%%%%%%%%%%%%%%%%%%%%%%%%%%%%%%%%%%%%%%%%%%%%%%%%%%%%%%%%%%%%%%%%%%%%%%%%%%%%%%%%%%%%%
%%%%%%%%%%%%%%%%%%%%%%%%%%%%%%%%%%%%%%%%%%%%%%%%%%%%%%%%%%%%%%%%%%%%%%%%%%%%%%%%%%%%%%%%%%%%%%%%%%%%%%%%%%%%%%%%%%%%%
%%
\section{Topology Detection Algorithm}
\label{sec:algorithm_statement}

%In Section \ref{sec:propagation_action_switches}, we showed how each switching action can be inferred from the voltage measurement in the network and how it is fully characterized by the eigenvector introduced in \eqref{eq:eigvct}. 

If we assume the distribution network physical infrastructure known, i.e. conductor impedances and switch locations, we can build a  \emph{library} $\mathcal L$ in which we collect all the normalized products between $\Lambda$ and the eigenvectors \eqref{eq:eigvct} for all possible breaker actions
\begin{equation}
\mathcal L = \left\{ g_{w_{-\ell}} : w \text{ satisfies Assumption \ref{ass:connection}} \right\}
\label{eq:library}
\end{equation}
where
\begin{equation}
g_{\sigma(t)_{-\ell}} = \frac{\Lambda \hat g_{\sigma(t)_{-\ell}}}{\|\Lambda \hat g_{\sigma(t)_{-\ell}}\|}
\label{eq:wgvct_lib}
\end{equation}
It is natural to compare at each time the trend vector $\delta(t,t-1)$ with the elements in the library to identify if and which switch changed its status. As stated in Remark \ref{rem:2}, if we want to identify which is the current topology we need additional information, i.e. the knowledge of the topology before the transition.
In that case, we could compare the trend vector with a restricted portion of the library $\mathcal L$, since there are only $\noswitches$ possible transitions, each of one caused by the action of one of the $\noswitches$ breakers. As a consequence, we can compare $\delta(t,t-1)$ with the \emph{particular library}
\begin{equation}
\mathcal L_{\sigma(t-1)} = \{ g_{\sigma(t)_{-\ell}} : \sigma(t)_{-\ell} = \sigma(t-1)_{-\ell} \}
\label{eq:part_library}
\end{equation}
that is peculiar of the state $\sigma(t-1)$ before the transition. 
The comparison is made by projecting the normalized measurements-based trend vector $\frac{\delta(t,t-1)}{\| \delta (t,t-1) \|}$ onto the topology library $\mathcal L_{\sigma(t-1)}$. The projection is performed with the inner product, and it allows us to obtain for each vector in  $\mathcal L_{\sigma(t-1)}$
\begin{equation}
c_{\sigma(t)_{-\ell}} = \left \| \left \langle \frac{\delta}{\| \delta \|}, g_{\sigma(t)_{-\ell}} \right \rangle \right \|.
\label{eq:proj}
\end{equation}
If $c_{\sigma(t)_{-\ell}} \simeq 1$, it means that $\delta$ is spanned by $g_{\sigma(t)_{-\ell}}$ and then that the switch $\ell$ changed its status. Because of the approximation \eqref{eq:approximate_solution}, the projection will never be exactly one. Therefore, we will use a heuristic threshold, called \emph{min{\textunderscore}proj}, based on our numerous simulations to select the right breaker.
If the projection is greater than the threshold, the associated switch is selected and the topology change time is detected. If there is no switches action, the trend vector will be zero as all the  $c_{\sigma(t)_{-\ell}}$, and the algorithm will not reveal any topology transition.
Thus, the projection value is used by the algorithm to detect the change time too, differently of what proposed in \cite{cavraroISGT2015}, where we used the norm of a matrix built by measurements (the \emph{trend matrix}). The new approach is more reliable in the realistic case. 
With a slight abuse of notation, we will say that the maximizer of $\mathcal C$, denoted by $\max \mathcal C$, is the switches status $\sigma$ such that 
$[\sigma]_{-\ell} = \sigma(t)_{-\ell}$, $[\sigma]_{\ell} = 1$ if $\sigma(t)_{-\ell}=0$ or vice-versa $[\sigma]_{\ell} = 0$ if $\sigma(t)_{-\ell}=1$ and $c_{\sigma(t)_{-\ell}}$ its the maximum element in $\mathcal C$.

We tacitly assumed so far all the buses endowed with a PMU, just to show the main idea. But this is not a realistic scenario for a distribution network. Now we show how the former approach can be generalized in presence of limited information, in which we are allowed to take only a few voltage measures: 
\begin{align}
y &= I_{\mathcal P} u \notag \\
&= I_{\mathcal P} \green_{\sigma(t)} \frac{\bar s }{U_N} + \1 U_N
\label{eq:fewmeas}
\end{align}
where $I_{\mathcal P} \in [0,1]^{\nosensors \times \nonodes}$ is a matrix that select the entries of $u$ where a PMU is placed.
In that case, the trend vector becomes
\begin{equation}
\delta(t_1,t_2) = y(t_1) - y(t_2)
\label{eq:trendvect_fewpmu}
\end{equation}
The library vectors and their dimension change too. In fact one can easily show, using \eqref{eq:fewmeas} and retracing \eqref{eq:ux2} and \eqref{eq:PhiMiller} that \eqref{eq:wgvct_lib} becomes
\begin{equation}
g_{\sigma(t)_{-\ell}} = \frac{I_{\mathcal P}\Lambda\hat g_{\sigma(t)_{-\ell}}}{\|I_{\mathcal P}\Lambda \hat g_{\sigma(t)_{-\ell}}\|}
\label{eq:egvct_lib}
\end{equation}
The topology detection algorithm with limited measurements is stated in Algorithm \ref{alg:DetAlg}.
\begin{small}
\begin{algorithm}
\caption{Topology Changes Detection}
	\begin{algorithmic}[1]
    \Require At each time $t$, $\sigma(t-1)$, \emph{min{\textunderscore}proj} = 0.98  
    	%\Require $s$
		\State $\sigma(t) \leftarrow \sigma(t-1)$ 
		
		\State each PMU at each node $j$ record voltage phasor measurements $y_j(t)$
		\State the algorithm builds the trend vector $\delta(t,t-1)$
		\State the algorithm projects $\delta(t,t-1)$ in the library $\mathcal L_{\mathcal P,\sigma(t)} $ obtaining the set of values $$\mathcal C = \left \{c_{\sigma(t)_{-\ell}} = \left \| \left \langle \frac{\delta}{\| \delta \|},g_{\sigma(t)_{-\ell}} \right \rangle \right \|, g_{\sigma(t)_{-\ell}} \in \mathcal L \right \};$$
		\If{$\max \mathcal C \geq \emph{min{\textunderscore}proj}$ } 
			\State $\sigma(t)  \leftarrow \arg \max \mathcal C$
		\EndIf
		\end{algorithmic}
\label{alg:DetAlg}
\end{algorithm}
\end{small}

\section{PMUs placement}
\label{sec:placement}
When we have a limited number of sensors to be deployed in the distribution grid, the first requirement to be satisfied is the system \emph{observability}, i.e.  the algorithm capability of detecting every topology transition.
Since our algorithm is basically based on the comparison between trend vectors and the library vectors, the trivial condition for the observability of the network is that each vector of the library is not proportional to any of the others. The former property is equivalent to the following condition.
\begin{lemma}
Given the set of nodes endowed with PMUs $\mathcal P$, let, with a slight abuse of notation, the juxtaposition of the library vectors be denoted by $\mathcal L$. Let $\mathcal L_{\mathcal P} = I_{\mathcal P} \mathcal L$. Then if 
\begin{equation}
(|\mathcal L_{\mathcal P}^* \mathcal L_{\mathcal P}|)_{uv} < 1, \forall u \neq v
\label{eq:lemm1_obsv}
\end{equation}
the switch that changes its status can be identified.
\label{lem:obsv1}
\end{lemma}
Notice that the element $(\mathcal L_{\mathcal P}^*\mathcal L_{\mathcal P})_{uv}$ is simply the projection of $g_{[\sigma(t)]_{-u}}$ onto $g_{[\sigma(t)]_{-v}}$. If they are note purely proportional, trivially the magnitude of their inner product is smaller than one. 
Condition \eqref{eq:lemm1_obsv} can be too restrictive, if we know the switches status before the topology change. In that case, in fact, we can just check if each vector of the \emph{particular} library is not proportional to any of the others. The former property is equivalent to the following condition.
\begin{lemma}
Let the switches status be $\sigma$, and let it be known. Given the set of nodes endowed with PMUs $\mathcal P$, let, with a slight abuse of notation, the juxtaposition of the particulary library vectors be denoted by $\mathcal L_{\sigma}$. Let $\mathcal L_{\mathcal P,\sigma} = I_{\mathcal P} \mathcal L_{\sigma}$. Then if 
\begin{equation}
(|\mathcal L_{\mathcal P,\sigma}^* \mathcal L_{\mathcal P,\sigma}|)_{uv} < 1, \forall u \neq v, \forall \sigma
\label{eq:lemm2_obsv}
\end{equation}
each switch action can be identified.
\label{lem:obsv2}
\end{lemma}

Lemmas \ref{lem:obsv1} and \ref{lem:obsv2} can be used to infer, given the electrical grid model, which is the minimun number of PMUs to be deployed in order to have the observability. 

A second issue is to find an ``optimal'' placement. %The notion of optimal is always relative to objective function to be maximized. 
%The paper main scope is to propose this novel approach for switches status, and thus we haven't already analyzed deeply the problem of optimal placement, neither how can it be stated as a formal optimization problem, that could eventually solved in an efficient way. 
%However, 
Now we propose a simple, even if onerous, strategy for placement. 
After finding a minimal number of sensors and a place that guarantee the satisfaction of Lemma \ref{lem:obsv1}, we propose a greedy PMU placement procedure based on the sequential addition of one PMU at a time able to provide the best performance improvement, verified by Monte Carlo simulations. 
For every possible new place for PMU, we run num{\textunderscore}run Monte Carlo simulation, of length TSTOP. For each of them we choose, randomly, the initial $\sigma$, the switch $\ell$ that changes its status and the time $\tau$ of the action. The place for the new PMU that performs the minimum number of errors is then chosen.

\begin{algorithm}
\caption{Greedy PMU placement}
    \begin{algorithmic}[1]
      \Require $\mathcal P $ such that the network is observable, num{\textunderscore}run, TSTOP 
			\State min $\leftarrow \infty$
			\State minimizer $\leftarrow \emptyset$
				\For{every possible place $j \notin \mathcal P$}
					\State err $ \leftarrow 0$
					\For{$t = 1$ to num{\textunderscore}run}
						\State choose $\sigma \sim \mathcal U([0,1]^\noswitches)$
						\State choose $\ell \sim \mathcal U([0,\dots,\noswitches])$
						\State choose $\tau \sim \mathcal U([0,\dots,\text{TSTOP}])$
						\State run Monte Carlo simulations
						\State err = number of errors
						\If{err $\leq$ min} 
							\State $\min \leftarrow $ err
							\State minimizer $\leftarrow j$
						\EndIf
					\EndFor
				\EndFor
			\State $\mathcal P \leftarrow \mathcal P \cup \{\text{minimizer}\} $ 
    \end{algorithmic}
\label{alg:greedy}
\end{algorithm}
\section{Switching action in the non-ideal scenario}
\label{sec:noisyscen}
So far, we considered the case in which the measurements devices were not affected by noise and the loads were static, which is not a realistic case.

The measurement apparatus in a bus where measurements are taken is formed by a PMU and by a potential transformer (PT). A PT for metering aims to reduce the voltage magnitude, in order to make it measurable by a PMU. Both the PMU and the PT introduce errors to the voltage measured. We model the output of our PMU placed at bus $j$ at time $t$ by 
\begin{small}
\begin{equation}
y_j(t) = u_j(t) + e_j(t) + b_j(t)
\label{eq:PMUmeasurements}
\end{equation}
\end{small}
where $e_j \in \complexnumbers$ is the error introduced by the PMU, while $b_j \in \complexnumbers$ is the errors introduced by the PT. 
A common characterization of the error is the \emph{total vector error} (TVE). For example if $x$ is the variable to be measured, and $x^N$ is the measured value, the TVE is 
\begin{equation}
\text{TVE} = \frac{|x - x^N|}{|x|}
\label{eq:tve}
\end{equation}
% 
%We model the module of $e_j$ as a Gaussian random variable $|e_j| \sim \mathcal N (0,\sigma^2_e)$, with the standard deviation $\sqrt{\sigma^2_e} = \frac{TVE_M}{3}$, and the phase of $e_j$ as a uniform random variable $\angle e_j \sim \mathcal U(0,360)$. In this way, with a probability of 99\%, the measures will be affected by a maximum $TVE$ of $TVE_M$. 

Furthermore, the loads are not static but they have a natural dynamic. 
In this paper we assume that the loads have\emph{ constant power factor}, and consequently 
\begin{equation}
\frac{q(t)_j}{p(t)_j} = \gamma_j , \quad \forall t,\;j = 2, \dots, \nonodes.
\end{equation}
We model the active power and the reactive power consumption at each load by
\begin{align}
p(t+1)_{-1} &= p(t)_{-1} + n_p(t) \label{eq:deltap}\\
q(t)_{-1} &= \text{diag}(\gamma_2,\dots,\gamma_\nonodes) p(t)_{-1}
\label{eq:deltaq}
\end{align}
where $n_p(t)$ is a Gaussian random vector. 
This could seems too rough and unsophisticated, but we will see in section \ref{sec:loadchar} that is actually enough accurate. 

If we take into account \eqref{eq:PMUmeasurements}, \eqref{eq:deltap} and \eqref{eq:deltaq}, the trend vector becomes
\begin{align}
\delta(t_1,t_2) &=   I_{\mathcal P}(\green_{\sigma(t_1)} - \green_{\sigma(t_2)}) \frac{\bar s (t_2)}{U_N} + e_{t_1} - e_{t_2} + \notag \\
& + b_{t_1} - b_{t_2} + \frac{I_{\mathcal P} \green_{\sigma(t_1)}}{U_N} \sum_{t = t_2}^{t_1-1} n_p(t) - i n_q(t),
\label{eq:noisy_trendvect}
\end{align}
being 
%\begin{align*}
$ p(t_1)_{-1} = p(t_2)_{-1} + \sum_{t = t_2}^{t_1-1} n_p(t). $%\\
%q(t_1) &= q(t_2) + \sum_{t = t_2}^{t_1-1} n_q(t)
%\label{eq:deltaq}
%\end{align*}
%
Since a PT is a passive device and since in the following we will consider only fast measurement rate, we approximate satisfactorily the errors introduced by a PT with a constant, i.e. we will assume that $b_j(t) = b_j$. As a consequence, the term $b_{t_1} - b_{t_2}$ in \eqref{eq:noisy_trendvect} vanishes.
Because of measurements noise and loads dynamic, the trend vector is typically non-zero even if there has not been any switching action, and
 the projection \eqref{eq:proj} may be almost one, leading to false topology detection.
When any switch is closed or opened, since we are adding or deleting a branch, we are changing the currents flows, reflecting on an abrupt, greater voltages variation. Therefore, a first strategy to avoid false positive is to not consider trend vector whose norm is lower than a defined threshold, called in the following \emph{min{\textunderscore}norm}.
Moreover the additive noise can make the maximum projection value of the trend vector $\max \mathcal C$ onto the library vectors considerably lower than one, even if a topology change occurred. This fact prompts us to use a value for \emph{min{\textunderscore}proj} lower than the one considered in Algorithm~\ref{alg:DetAlg}.
Of course, the use of a lower threshold makes the algorithm more vulnerable to false positive. 
The following example will give the idea for a possible solution.
\begin{example}
\label{ex:es3}
Assume the grid without load variation and measurements noise, and that at time $t_1$ the $\ell$-th switch change its status. Consider the trend vector 
$$\delta(t,t-\tau) = y(t) - y(t-\tau).$$
For $t<t_1$ and $t \geq t_1 + \tau$ the projections of the trend vector onto the library are all equal to zero, because
$$\delta(t,t-\tau) = y(t) - y(t-\tau) = 0$$
Instead for $t_1 \leq t < t_1 + \tau$, the trend vector is
$$\delta(t,t-\tau) = \Lambda \Phi_{\sigma(t)_{-\ell}} \Lambda^T \frac{\bar s }{U_N}$$
leading to a cluster of algorithm time instant of length $\tau$ (or $\frac \tau f$ seconds),   in which the maximum projection coefficient will be almost one.
\end{example}

A possible solution is thus to consider a trend vector built using not two consecutive measures, but considering measures separated by $\tau$ algorithm time istants , i.e. 
$$\delta(t,t-\tau) = y(t) - y(t-\tau).$$
and to assume that a topology change has happened at time $t$ when we have a cluster whose length {length{\textunderscore}cluster} is $\tau$ of consecutive values of $\max \mathcal C$ greater than {min{\textunderscore}proj}.
This idea will be clarified with some simulations in Section \ref{sec:results}. 
The former observations lead to the Algorithm \ref{alg:DetAlg_noisy} for topology detection with measurements noise and load variation.
\begin{small}
\begin{algorithm}
\caption{Topology transition detection with noise}
    %
    	%\Require $s$
	\begin{algorithmic}[1]
  \Require At each time $t$, we are given the variables $\sigma(t-1)$, minimizer$(t-1)$, length{\textunderscore}cluster$(t-1)$  
	\State $\sigma(t) \leftarrow \sigma(t-1)$ 
		
	\State each PMU at each node $j$ record voltage phasor measurements $y_j(t)$
	\State the algorithm builds the trend vector $\delta(t,t-\tau)$
	\If{$\|\delta(t,t-\tau)\| <$ min{\textunderscore}norm} 
			\State $\delta(t,t-\tau)  \leftarrow 0$
			\State $\text{minimizer}(t) = 0$
			\State length{\textunderscore}cluster$(t)=0$
	\Else
			\State the algorithm projects $\delta(t,t-\tau)$ in the particular library $\mathcal L_{\mathcal P,\sigma(t)} $ obtaining the set of values $$\mathcal C = \left \{c_{\sigma(t)_{-\ell}} = \left \| \left \langle \frac{\delta}{\| \delta \|},g_{\sigma(t)_{-\ell}} \right \rangle \right \|, g_{\sigma(t)_{-\ell}} \in \mathcal L \right \};$$
			\If{$\max \mathcal C > $min{\textunderscore}proj } 
					\State $\text{minimizer}(t) = \arg \min \mathcal C$
					\If{$\text{minimizer}(t) = \text{minimizer}(t-1)$ } 
					\State length{\textunderscore}cluster$(t) \leftarrow$ length{\textunderscore}cluster$(t-1)$ + 1
							\If{length{\textunderscore}cluster$(t) = \tau$}
									\State $\sigma(t) \leftarrow \text{minimizer}(t) $
							\EndIf
					\Else 
					\State length{\textunderscore}cluster$(t) \leftarrow 1$
					\EndIf
			\EndIf
	\EndIf
	\end{algorithmic}
\label{alg:DetAlg_noisy}
\end{algorithm}
\end{small}

\section{Dynamic characterization of actual loads}
\label{sec:loadchar}
Load behavior at the individual customer level and at high time resolution can be a critical question in distribution networks. Due to lack of accurate and high resolution measurements at meters, there is no a clear answer to that question. The most commonly available data from loads come from meters with hourly (or 15 minutes) time intervals. 
For applications that are based on network parameters in shorter time than hourly, large uncertainty is caused by low resolution load data. 
For our topology detection algorithm, the load variation is very important, since it affects the topology detection capability, as \eqref{eq:noisy_trendvect} shows. To add more practicality to the proposed topology detection algorithm, a load measurement data set for five houses in the United States is used. 
Load demand (kW) is recorded every second for a week. Statistical analysis of these load data is presented in Table~\ref{tab:loadDataStat1}. 
\begin{figure}
\centering
\includegraphics[width=0.45\textwidth]{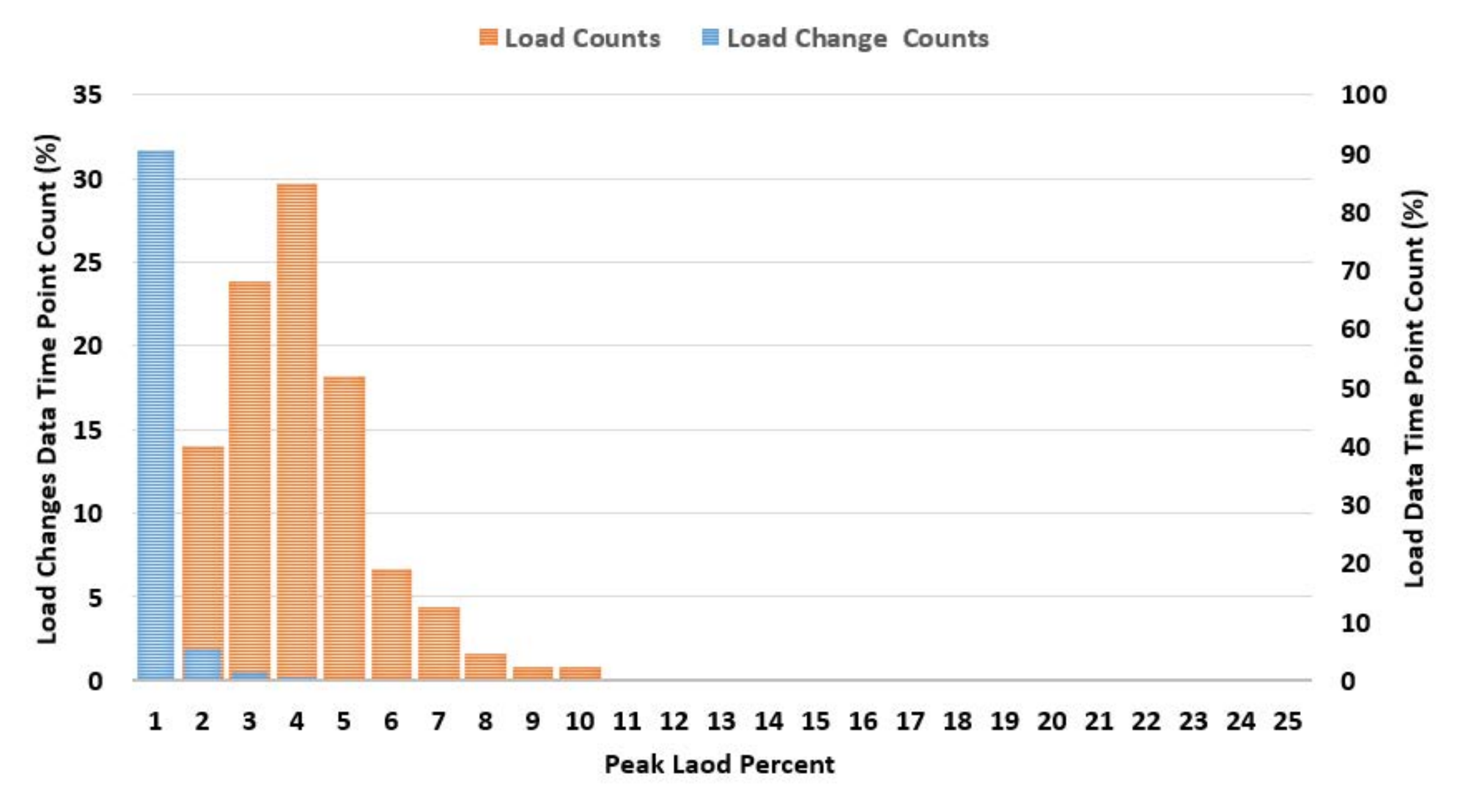}
\caption{Load variability and load change variability based on normalized load values.}
\label{fig:LoadDataHistogram}
\end{figure}	
Table~\ref{tab:loadDataStat2} reports instead the statistical analysis of the variation of the load demand between two consecutive seconds. The field ``Relative SD'' expresses the standard deviation of the load or of the load change as a percentage of the load peak value, written in Table~\ref{tab:loadDataStat1}.
\begin{table}
\caption{Load values}
\label{tab:loadDataStat1}
\centering
\begin{tabular}{lcccc}
           & SD (kW) & Max (kW) & Relative SD (\%) 	\\
House 1    & 0.415   & 5.990    & 6.93 \\
House 2    & 1.309   & 10.956		& 11.95\\
House 3    & 1.955   & 11.578		& 16.88\\
House 4    & 2.566   & 12.364		& 20.75\\
House 5    & 1.309   & 8.155		& 16.06\\
Aggregate  & 4.204   & 27.229   & 15.44
\end{tabular}
\end{table}
\begin{table}
\caption{Load variation}
\label{tab:loadDataStat2}
\centering
\begin{tabular}{lccc}
           & Mean (kW) & SD (kW) & Relative SD (\%) \\
House 1    & 0.000     & 0.045   & 0.11 \\
House 2    & 0.000     & 0.070   & 0.64 \\
House 3    & 0.000     & 0.113   & 0.98 \\
House 4    & 0.000     & 0.110   & 0.89 \\
House 5    & 0.000     & 0.046   & 0.56 \\
Aggregate  & 0.000     & 0.184   & 0.68 
\end{tabular}
\end{table}
Figure \ref{fig:LoadDataHistogram} shows the histogram of load duration and load changes. The orange bars represents measurements data points over the percentage of pick load values. The blue bars show load change in one second intervals over the percentage of pick load values. In case of load changes, more than 90 percent of measurement time points have loads with less that one percents of pick load. It means that during one second time intervals, there is not huge difference in load data. 
In the United States, a number of houses are connected to one distribution transformer. Therefore, the aggregated five houses loads are considered as the reference for load variability in this paper. 
The lower the measurements frequency, the higher the load variability. Resampling  the data, we obtain the aggregated characterization reported in Table \ref{tab:loadDataStat3} for different measurements frequencies, i.e. for one measure every second, one measure every five seconds and one measure every ten seconds.
\begin{table}
\caption{Load variation for differents frequency}
\label{tab:loadDataStat3}
\centering
\begin{tabular}{lccc}
           & Mean (kW) & SD (kW) & Relative SD (\%) \\
$f=1\text{ Hz}$  	& 0.000     & 0.184   & 0.68\\
$f=0.2\text{ Hz}$	& 0.000     & 0.425   & 1.56  \\
$f=0.1\text{ Hz}$	& 0.000     & 0.604   & 2.22 
\end{tabular}
\end{table}
Further numerical analysis shows that the load differences, for all the frequencies, can be modeled as Gaussian random variables thus validating equations \eqref{eq:deltap} and \eqref{eq:deltaq}, that can be assumed as a realistic model of load variation. 
%Figure \ref{fig:normplot} is the output, for the 1 second data, of the MATLAB function \texttt{normplot}, whose purpose is to graphically assess whether the data could come from a normal distribution. More similar to the normal distribution, more linear the plot will be. We can see how the plot is almost linear.
%
%\begin{figure}
%\centering
%\includegraphics[width=0.32\textwidth]{figures/normplot1sec}
%\caption{Output of \texttt{ normplot} for the 1 second data}
%\label{fig:normplot}
%\end{figure}	
%

%%%%%%%%%%%%%%%%%%%%%%%%%%%%%%%%%%%%%%%%%%%%%%%%%%%%%%%%%%%%%%%%%%%%%%%%%%%%%
%%%%%%%%%%%%%%%%%%%%%%%%%%%%%%%%%%%%%%%%%%%%%%%%%%%%%%%%%%%%%%%%%%%%%%%%%%%%%

\section{Results, Discussions and Conclusions}
\label{sec:results}
We tested our algorithm for topology detection on the IEEE 33-bus distribution test feeder \cite{parasher2014load}, which is illustrated in the Figure \ref{fig:ieee33}.
In this testbed, there are five breakers (namely $S_1$, $S_2$, $S_3$, $S_4$, $S_5$) that can be opened or closed, thus leading to the set of 32 possible topologies $\topology_1,\dots,\topology_{32}$. Because of the ratio between the number of buses and the number of switches, some very similar topologies can occur (for example the topology where only $S_1$ is closed and the one in which only $S_2$ is closed). 
In the IEEE33-bus test case, Assumption \ref{ass:sameR/X} about line impedances does not hold, making the test condition more realistic. Each bus of the network represent an aggregate of five houses, whose power demand is described by the statistical Gaussian model derived in Section \ref{sec:loadchar}, dependent on the sampling frequency, and whose characterization is reported in Table \ref{tab:loadDataStat3}.  
Regarding the measurement noise, we assume that the buses are endowed with high precision devices, the $\mu$PMU \cite{microPMU}, and with PTs. PMUs measurements are effected by Gaussian noise such that $\text{TVE} \leq 0.05 \%$, based on the PMU manufacturer test information. It is also comply the IEEE standard C37.118.1-2011 for PMUs \cite{6111219}. The PTs instead introduce a constant error (but peculiar of each node) that satisfy the requirements of the standard \cite{std199760044}.
\begin{figure}[]
\centering	
\includegraphics[width=0.28\textwidth]{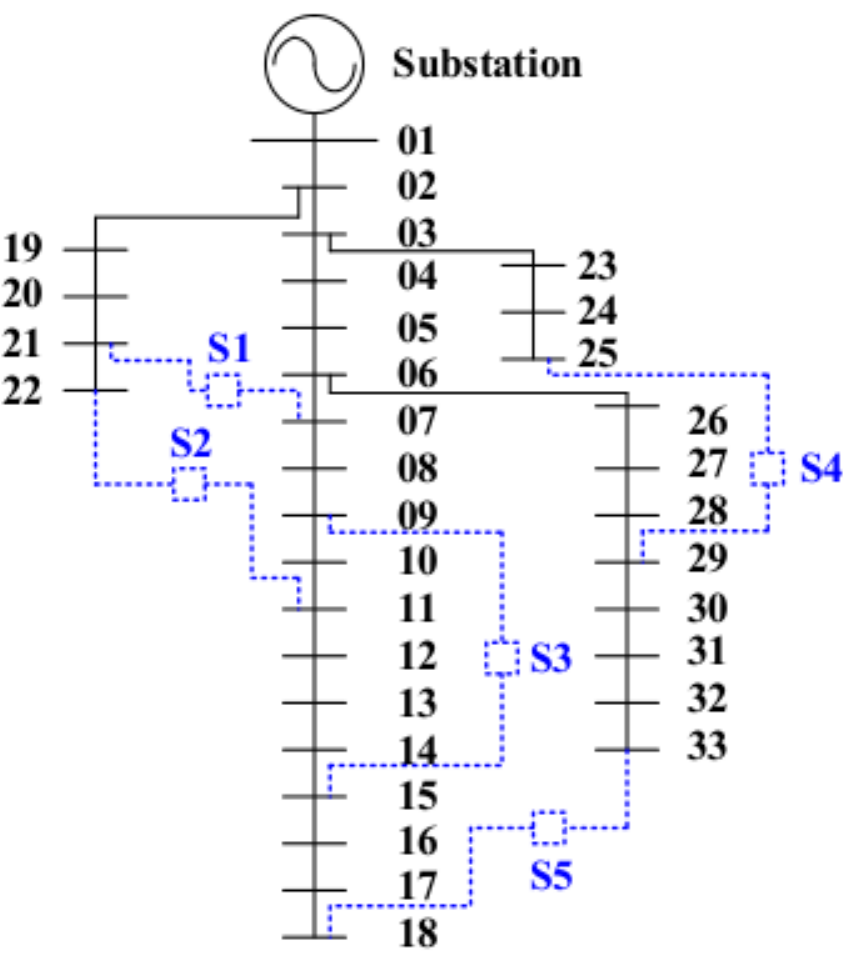}
\caption{Schematic representation of the IEEE33 buses distribution test case with the five switches}
\label{fig:ieee33}
\end{figure}

\subsection{Trend vectors and noise treatment strategy}
Here we provide simulations that show the problems given by noise and that validate the strategy we use to overcome these problems. The overall time window simulated is of 1000 seconds, the measures frequency is of one every 10 seconds and the topology transition ($f = 1/10 \text{ sec}^{-1}$), from $\sigma_1 = (1,1,1,0,1)$ to $\sigma_2 = (1,1,0,0,1)$, happens at $t = 480$ sec. In Figure \ref{fig:normtrendvect_wn} we see what happens to the trend vector norm, while in Figure \ref{fig:projwn} we plot $\max \mathcal C$, when we are in a noiseless scenario with loads not time varying. We can see that the switch action instant is clear.
\begin{figure}
\centering	
\includegraphics[width=0.46\textwidth]{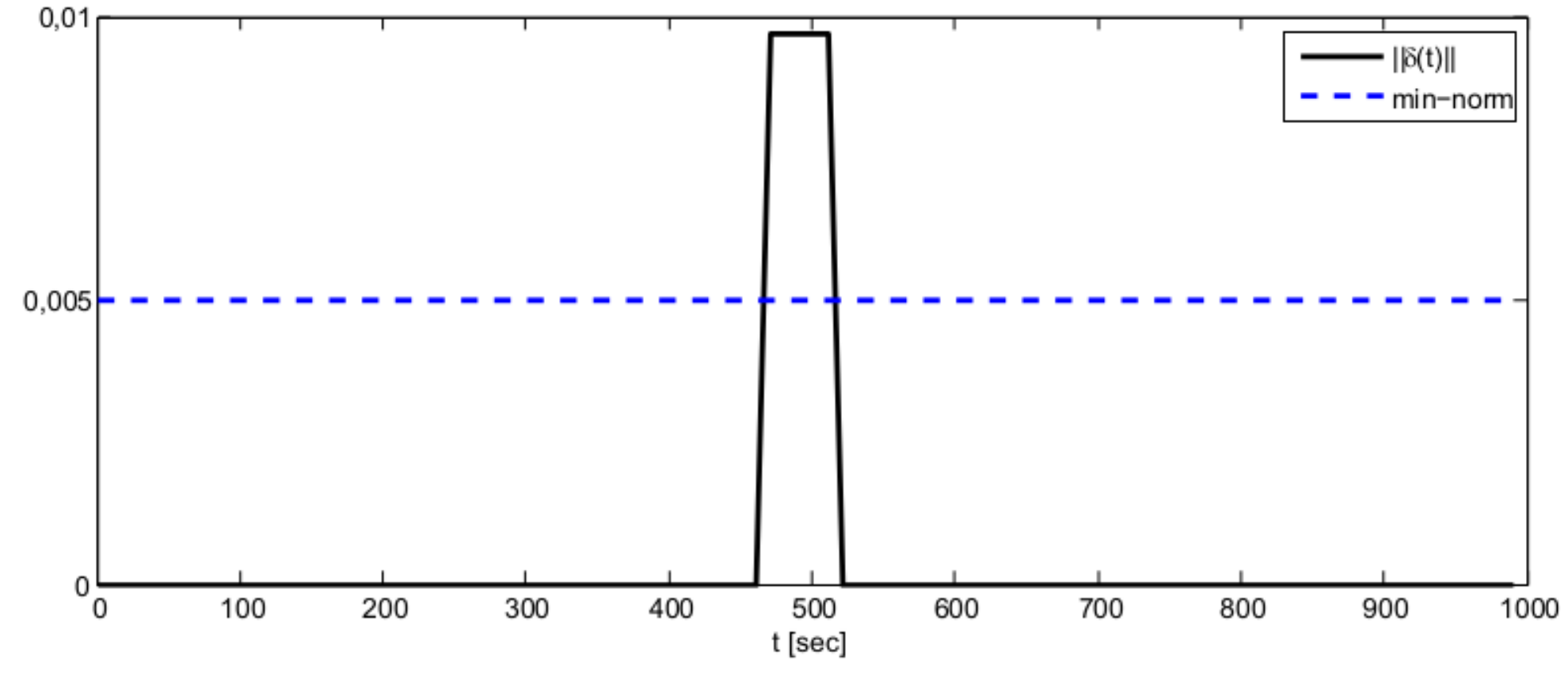}
\caption{Trend vector norm with noiseless measurements and loads non time varying}
\label{fig:normtrendvect_wn}
\end{figure}
\begin{figure}
\centering	
\includegraphics[width=0.48\textwidth]{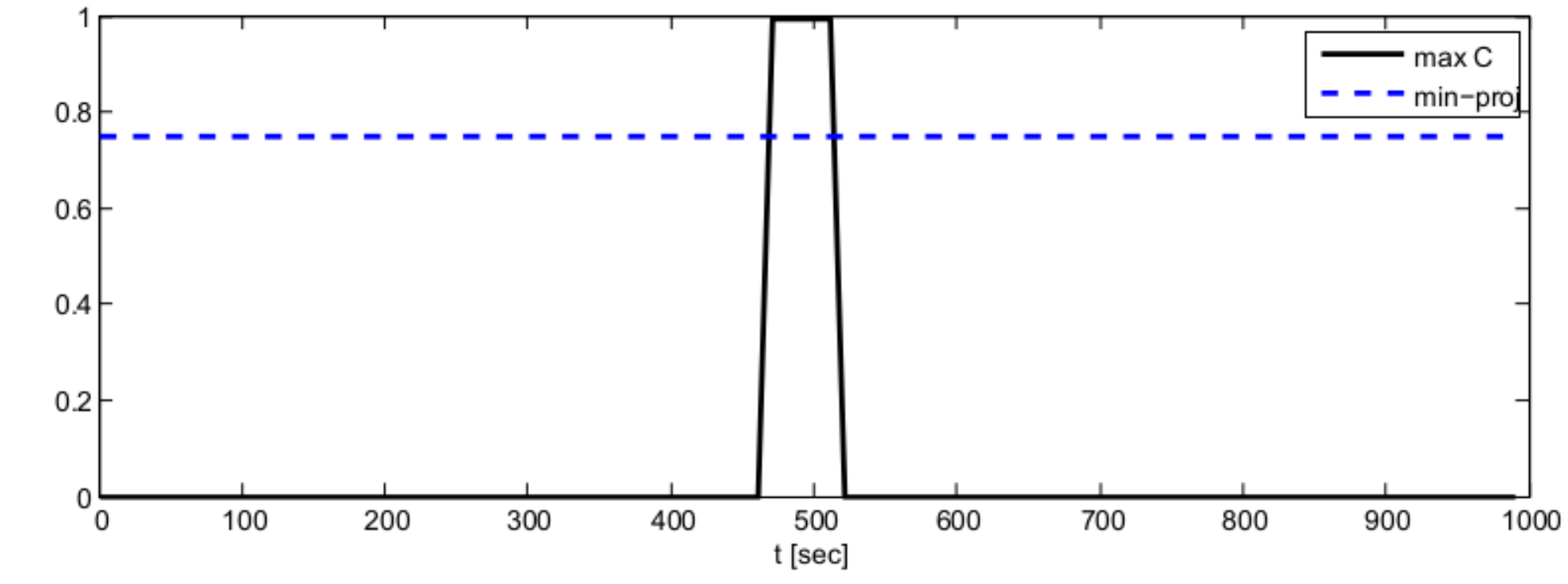}
\caption{Trajectory of the $\max \mathcal C$ with noiseless measurements and loads non time varying}
\label{fig:projwn}
\end{figure}
In the case with measurements noise and time varying loads, the trajectory of the trend vector norm, compared with \emph{min{\textunderscore}norm} numerically setted to $0.05$, is reported in Figure \ref{fig:normtrendvect}. In Figure \ref{fig:projnwn} instead we plot $\max \mathcal C$, without putting to zero the trend vectors whose norm is lower than \emph{min{\textunderscore}norm} before the projection. We see that there are several cluster of time istants in which $\max \mathcal C$ is greater than \emph{min{\textunderscore}norm}, and this leads to a number of faulse positive. 
In Figure \ref{fig:projwnth} we plot $\max \mathcal C$, after putting to zero the trend vectors whose norm is lower than \emph{min{\textunderscore}norm} before the projection. We still have more than one cluster, but only one has a length of $\tau \frac 1 f$ seconds, thus revealing the switch action instant, and the validity of our strategy.
\begin{figure}
\centering	
\includegraphics[width=0.48\textwidth]{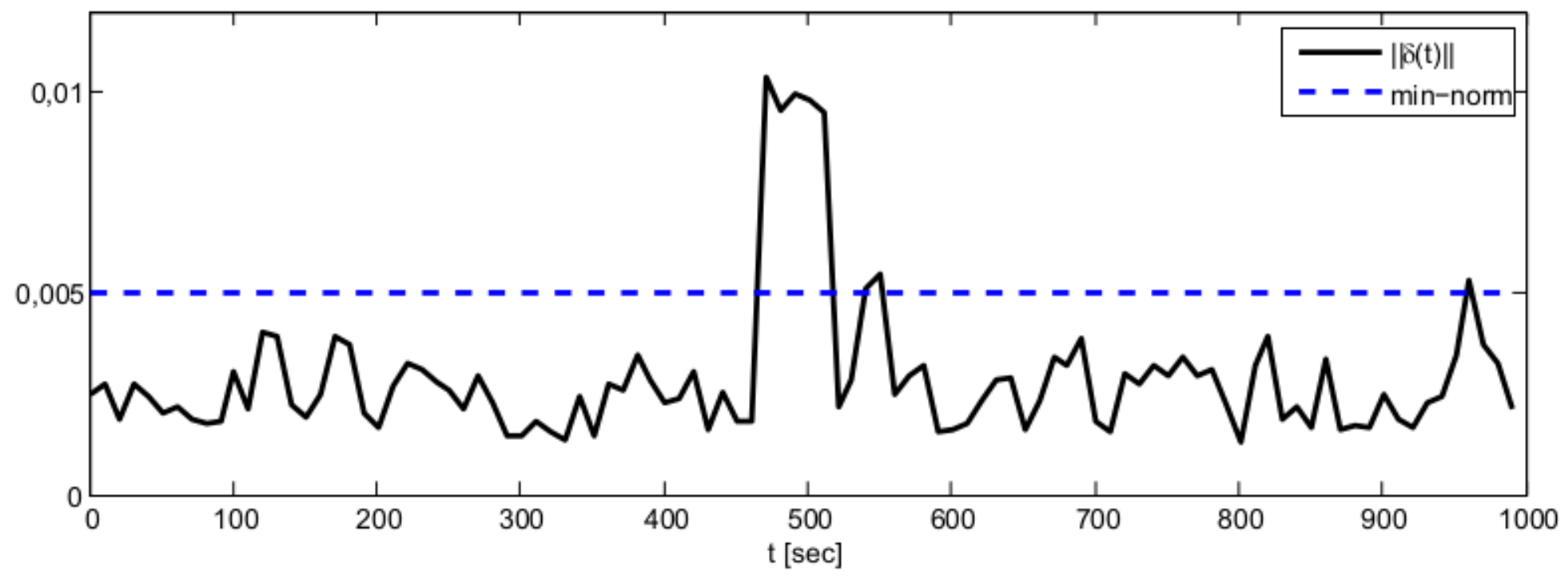}
\caption{Trend vector norm with noisy measurements and time varying loads}
\label{fig:normtrendvect}
\end{figure}
\begin{figure}
\centering	
\includegraphics[width=0.48\textwidth]{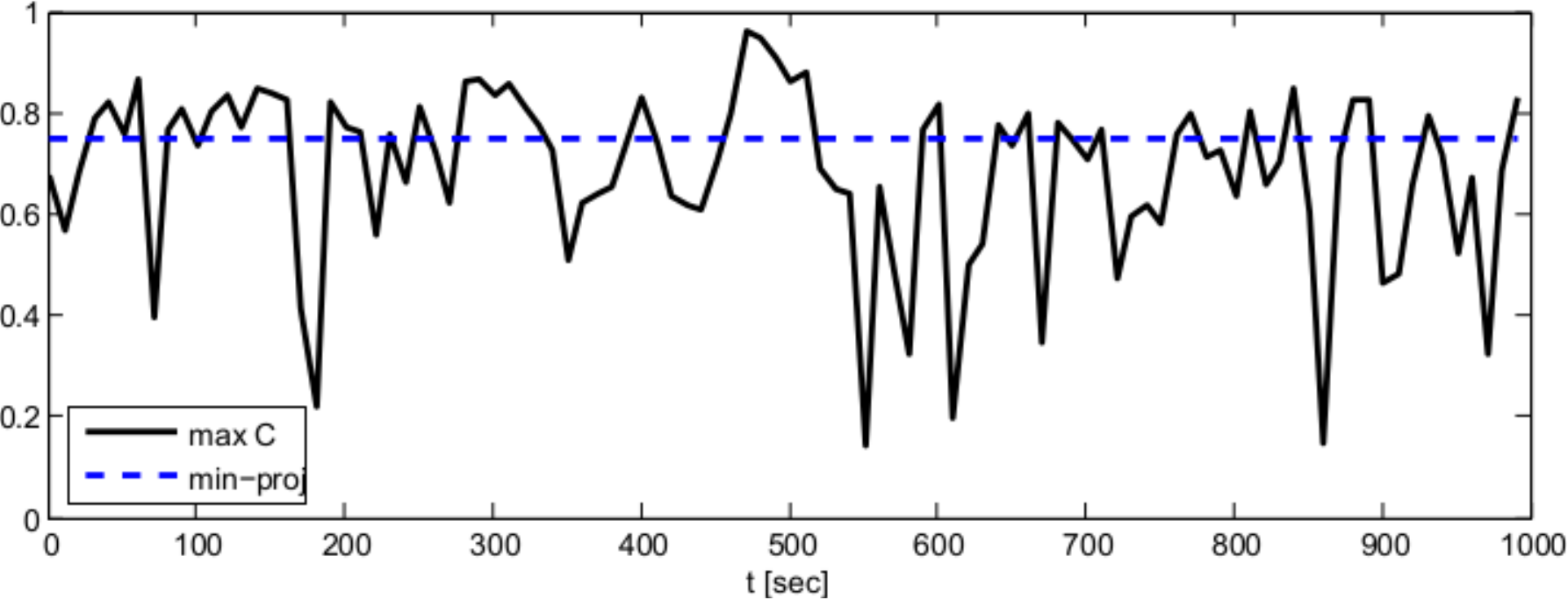}
\caption{Trajectory of the $\max \mathcal C$ with noisy measurements and time varying loads}
\label{fig:projnwn}
\end{figure}
\begin{figure}
\centering	
\includegraphics[width=0.48\textwidth]{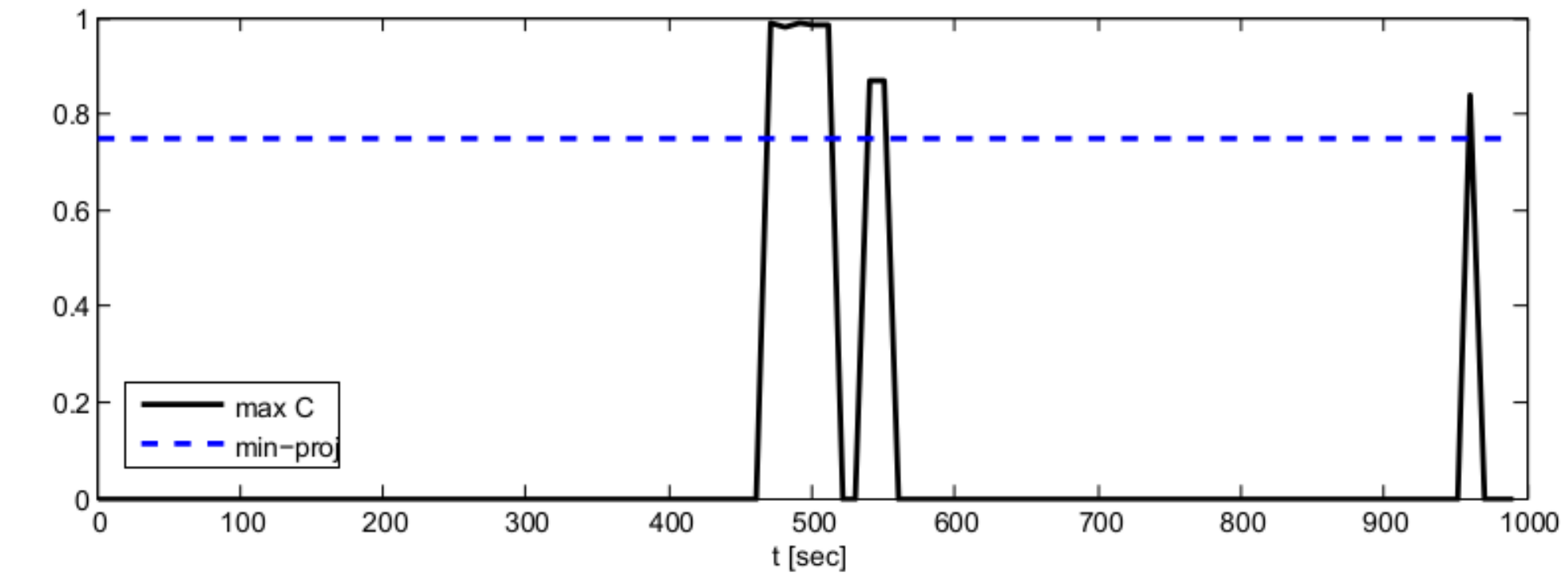}
\caption{Trajectory of the $\max \mathcal C$ with noisy measurements and time varying loads after putting to zero the vectors whose norm is lower than \emph{min{\textunderscore}norm}}
\label{fig:projwnth}
\end{figure}

\subsection{Simulation of the algorithm}

Here we tested the entire switches monitoring algorithm, in different situations and with different placements of PMUs:
\begin{itemize}
\item $\mathcal P_{33} = \{1,\dots,33 \}$, where every node is endowed with a PMU;
\item \begin{small}
$\mathcal P_{15} = \{3,8,9,10,12,15,16,17,18,19,21,24,25,27,30 \}$
\end{small}, where there are 15 PMUs, almost one every two nodes;
\item $\mathcal P_7 = \{9,12,15,18,24,27,30 \}$, where there are 7 PMUs, almost one every four nodes
\end{itemize}
$\mathcal P_{15}$ and $\mathcal P_7$ has been computed using Algorithm~\ref{alg:greedy}. The algorithm has been tested in each condition via 10000 Monte Carlo simulations..
Firstly we tested the algorithm in the ideal case without load variation nor measurement noise, resulting in no errors. Therefore, we can see that in the steady-state condition and in the absence of noise, the algorithm is extremely efficient for the 33-bus test case. It also overcomes the linearization from Proposition~\ref{pro:approximation} and the initial Assumption~\ref{ass:sameR/X}.
Secondly we added the measurement noise and different levels of load variation (or alternatively of measures frequency). The results are reported in Table~\ref{tab:sim33}, Table~\ref{tab:sim15} and Table~\ref{tab:sim7}. 
The field ``non detection'' refers to the number of run in which the algorithm doesn't comprehend that there has been a switching action, the field ``wrong detection'' refers to the number of times the algorithm detects a false action, while ``decision errors'' is the number of times the algorithm provide a wrong breakers status estimation.
Of course every time there is a wrong detection, we have a decision error, too. If we subtract the values of the second column to the values of the third, we find very small number, meaning that, once detected the exact action time, the algorithm works very well.  
The main challenge is therefore to detect the topology transition time. 
However our approach is very robust with the number of sensors: the percentage of errors with only 7 PMUs (where we have partial information on the grid state) is always worst than the one with 33 PMUs (where we have global information on the grid state) less than $2.5\%$, and probably with the best parameter tuning (\emph{min{\textunderscore}norm} and \emph{min{\textunderscore}proj}), it can still be improved.  
\begin{table}
\caption{Simulation Results with 33 PMUs}
\label{tab:sim33}
\centering
\begin{tabular}{lccccc}
SD [kV] & non & wrong  & decision & total & perc. of \\
	& detections & detection & errors & errors & errors (\%) \\
0  		& 0     & 50   	& 50    & 100	& 1.00 \\
0.68, ($f=1\text{ Hz}$)	& 0     & 64   	& 67    & 131 	& 1.31 \\
1.56, ($f=0.2\text{ Hz}$)	& 17		& 131   & 152 	& 300		& 3.00 \\
2.22, ($f=0.1\text{ Hz}$)	& 72   	& 211  	&	244		& 527		& 5.27 
\end{tabular}
\end{table}
\begin{table}
\caption{Simulation Results with 15 PMUs}
\label{tab:sim15}
\centering
\begin{tabular}{lccccc}
Relative & non & wrong  & decision & total & perc. of \\
SD (\%)	& detections & detection & errors &errors & errors (\%) \\
0 		& 0     & 52   	&  52  & 104		& 1.04 \\
0.68, ($f=1\text{ Hz}$)	& 0     & 73   	&  76	& 149 	& 1.49 \\
1.56, ($f=0.2\text{ Hz}$)	& 29			& 135   	&  158  & 322		& 3.22 \\
2.22, ($f=0.1\text{ Hz}$)	& 74    & 213  		&	 252	& 539		& 5.39 
\end{tabular}
\end{table}
\begin{table}
\caption{Simulation Results with 7 PMUs}
\label{tab:sim7}
\centering
\begin{tabular}{lccccc}
Relative & non & wrong  & decision & total & perc. of \\
SD (\%)	& detections & detection & errors &errors & errors (\%) \\
0 		& 0     & 56   	&  56  & 112		& 1.12 \\
0.68, ($f=1\text{ Hz}$)	& 2     & 180   	&  185	& 367 	& 3.65 \\
1.56, ($f=0.2\text{ Hz}$)	& 31			& 199   	&  209  & 441		& 4.41 \\
2.22, ($f=0.1\text{ Hz}$)	& 76    & 245  		&	 298	& 619		& 6.19 
\end{tabular}
\end{table}

\section{Conclusions}
\label{sec:conclusions}

In this paper we propose a novel strategy for the monitoring and the identification of switches action in a distribution grid. The novelty of this algorithm is the possibility of running it in real time and the fact that it could work satisfactorily with only a partial knowledge of the grid state. It allows us, exploiting a voltage phasorial measurements, to understand both the time of the switching action and the new topology, by comparing the trend vector, built by data, with other vectors contained in a library, that represent the a priori knowledge of the electrical network.
The algorithm has been tested in a realistic scenario, where both measurement noise and load variation had realistic characterization. In particular, we provide a simply but plausible model of the fast time scale load variation, validated using real field load data.
The simulations show the algorithm behavior in different scenario.
In the ideal one, with static loads and perfect measurements devices, the algorithm gives no errors. When instead we have noisy PMUs and load variation, we still have satisfactory results, strengthened by the fact that the algorithm performance are very similar both in the case of one PMU per bus or of one PMU every four nodes.
%\end{linenumbers}

\bibliographystyle{IEEEtran}
\bibliography{bibTex_Topology}

\end{document}